\documentclass{article}
\newcommand{\R}{{\mathbb R}}

\newcommand{\EE}{{\mathbb E}}

\newcommand{\de}{\mathrm{d}}

\newtheorem{thm}{Theorem}

\newtheorem{algo}{Algorithm}
\newenvironment{proof}{\noindent{\bf Proof:} }{\hfill $\square$ \\}

\usepackage[english]{babel}
\usepackage{amssymb}
\usepackage{amsmath}
\usepackage{epsf}
\usepackage{natbib}

\begin{document}
\pagestyle{myheadings}
\thispagestyle{empty}
\setcounter{page}{1}

\begin{center}
{\Large\bf
Non-parametric adaptive bandwidth selection for kernel estimators 
of spatial intensity functions}\\[0.5in]
{\large M.N.M.\ van Lieshout}\\[0.1in]
CWI \\
P.O. Box 94079, NL-1090 GB Amsterdam, The Netherlands \\[0.1in]
Department of Applied Mathematics,
University of Twente \\
P.O. Box 217, NL-7500 AE Enschede, The Netherlands\\[0.5in]
\end{center}

\begin{verse}
{\footnotesize
\noindent
{\em Abstract:} We propose a new fully non-parametric two-step 
adaptive bandwidth selection method for kernel estimators of 
spatial point process intensity functions based on the 
Campbell--Mecke formula and Abramson's square root law. We 
present a simulation study to assess its performance relative 
to the Cronie--Van Lieshout global bandwidth selector and 
apply the technique to data on induced earthquakes in the
Groningen gas field. \\[0.2in]
\noindent
{\em Key words \& Phrases:} 
adaptive kernel estimation;
bandwidth selection;
Camp\-bell--Mecke formula;
induced earthquakes;
intensity function;
point process.
}
\end{verse}

\section{Introduction}

The first step in any analysis of a spatial point pattern is usually 
estimating its intensity function \citep{Digg14,Illi08,Lies19}. To do 
so, various techniques exist. Perhaps the oldest is quadrat counting 
\citep{DuRi29} in which one simply reports the number of points falling 
in each quadrat scaled by the quadrat volume. Instead of fixed quadrats, 
one might use the cells in a tessellation formed by the pattern itself 
as the units for counting \citep{BarrScho10,Ord78,SchaWeyg00}. However, the 
most popular technique by far seems to be kernel estimation.

The classic kernel estimator for the spatial intensity function \cite{Digg85}
uses a constant bandwidth. However, intuitively, such a `one size fits all' 
approach would tend to over-smooth in dense areas, whilst not smoothing 
enough in sparse regions. As a consequence, finer detail in areas with 
many points may be lost, whereas the few points in sparser areas might 
give rise to spurious hot spots. 
Motivated by similar considerations for density estimation of one-dimensional 
random variables, \citet{Abra82a} proposed to use a kernel smoother in which 
the bandwidth at each observation is weighted by a power of the density at 
that observation. Doing so reduces the bias significantly \citep{Halletal95},
at least asymptotically. 

In kernel estimation, the crucial parameter is the bandwidth. It is often
chosen by visual inspection or using a rule of thumb (see e.g.\ 
\citet[Section~6.5]{Badd16}, \citet[Section~3.3]{Illi08} or 
\citet[Section~6]{Scott92}). Obviously, though, such procedures are 
rather ad-hoc and subjective. 

A different class of techniques is based on asymptotics. For instance for 
classic kernel estimators, \citet{BrooMarr91} considered a Poisson point 
process on the real line and assumed a simple multiplicative model for the 
intensity function to derive an asymptotically optimal least-squares 
cross-validation estimator when the number of points tends to infinity. 
\citet{Lo17} picked up the baton and studied the asymptotic (integrated) 
mean squared error in any dimension without imposing a specific intensity 
model, again in the regime that the number of points goes to infinity.
\citet{Lies20} generalised Lo's work to point processes that may exhibit
interaction between the points under the assumption that replicated patterns
are available so that infill asymptotics apply. \citet{DaviFlynHaze18} 
considered asymptotic expansions for a spatial analogue of the Abramson 
estimator for Poisson processes as the number of points increases; 
\citet{Lies21} studied infill asymptotics that allow for interaction between 
the points. It is important to note that the resulting optimal bandwidths
depend on the unknown intensity function and cannot be computed in practice 
without resorting to iterative techniques.

Less subjective yet practical procedures to select a suitable global bandwidth 
rely on a specific model. For example likelihood cross-validation 
\citep[Section 5.3]{Load99} assumes the data come from a Poisson point process. 
Another common approach is to minimise the mean squared error in state estimation 
for a planar stationary isotropic Cox process \citep{Digg85}. The disadvantage
of such techniques is that the underlying assumption may not hold for the 
pattern at hand, which motivated \citet{CronLies18} to propose a fully 
non-parametric technique. For adaptive bandwidth selection, to the best of our
knowledge, similar procedures do not exist. In this article, we extend the
Cronie--Van Lieshout approach to adaptive bandwidth selection and propose a 
new fully non-parametric, easy to implement, two-step adaptive bandwidth 
selection method based on the Campbell--Mecke formula that does not require
numerical approximation of integrals nor knowledge of second or higher moments.

The plan of this paper is as follows. Section~\ref{S:prelim} recalls crucial 
concepts and fixes notation. In 
Section~\ref{S:adaptive}, we discuss adaptive kernel estimators and present
the algorithm for selecting the bandwidth. The results of a simulation 
study into the efficacy of the new approach are given in Section~\ref{S:simu};
an application to a data set concerning induced earthquakes is presented
in Section~\ref{S:Groningen}. The paper closes with a discussion on 
computational complexity and ideas for future research.

\section{Preliminaries and notation}
\label{S:prelim}

First, let us introduce some notation. Let $\Psi$ be a simple point process 
\citep{SKM} in $d$-dimensional Euclidean space $\R^d$ that is observed in a bounded,
non-empty and open subset $W$ of $\R^d$. We assume that the first order moment 
measure $\Lambda$ of $\Psi$ defined by
\[
\Lambda(A) = \EE \left[ \sum_{x\in \Psi} 1\{ x \in A \} \right],
\]
the expected number of points of $\Psi$ that fall in Borel subsets $A$ of $\R^d$, 
exists as a locally finite Borel measure and is absolutely continuous with respect 
to $d$-dimensional Lebesgue measure $\ell$ with a Radon--Nikodym derivative $\lambda: 
\R^d \to [0,\infty)$. We will refer to the function $\lambda$ as the {\em intensity 
function\/} of $\Psi$. 

The {\em kernel estimator\/} of the intensity function of a point process was 
introduced by \citet{Digg85} as
\begin{equation}
\label{e:kernel}
\widehat \lambda(x_0; h, \Psi, W) = \frac{1}{h^d} \sum_{y\in\Psi\cap W} 
\kappa\left( \frac{x_0 - y}{h}\right),
\quad x_0 \in W,
\end{equation}
possibly divided by a global edge correction factor 
\[
w(x_0, h, W) = 
  \frac{1}{h^d} \int_W \kappa\left( \frac{x_0 - z}{h}\right) dz.
\]
An alternative, local, edge correction can be found in \citet{Lies12}. 
The function $\kappa: \R^d \rightarrow [0,\infty)$ is supposed to be a kernel, 
that is, a $d$-dimensional probability density function \citep[p.~13]{Silv86} 
that is even in all its arguments. When $\kappa$ is positive in a 
neighbourhood of the origin, since $W$ is assumed to be open, the global 
edge correction factor is non-zero for all $x_0 \in W$.

The crucial parameter in (\ref{e:kernel}) is the {\em bandwidth\/} $h>0$, which 
determines the amount of smoothing. For large $h$, the mass of $\kappa$ is spread 
far and wide, which reduces the variance but may lead to a large bias. For small 
$h$, the mass of $\kappa$ is concentrated around the observed points of $\Psi \cap W$. 
Thus, the bias is reduced at the price of a larger variance. 

Popular choices of kernel include those belonging to the Beta class \citep{Hall04}
\begin{equation}
\label{e:beta}
\kappa^\gamma(x) = 
\frac{\Gamma \left(d/2 + \gamma + 1\right)}{
\pi^{d/2} \Gamma \left(\gamma+1\right)}
(1 - x^T x)^{\gamma} \, 1\{ x \in B(0, 1) \},
\quad x\in\R^d,
\end{equation}
for $\gamma \geq 0$. Here $B(0,1)$ is the closed unit ball in $\R^d$ centred 
at the origin.  Note that Beta kernels are supported on the compact unit ball 
and that their smoothness is governed by the parameter $\gamma$. Indeed, the box 
kernel defined by $\gamma=0$ is constant and therefore continuous on the interior 
of the unit ball; the Epanechnikov kernel corresponding to the choice $\gamma=1$ 
is Lipschitz continuous. For $\gamma > k$ the function $\kappa^\gamma$ is $k$ times 
continuously differentiable on $\R^d$. An alternative with unbounded support is 
the Gaussian kernel
\begin{equation}
\label{e:gauss}
\kappa(x) = (2\pi)^{-d/2} \exp\left( - x^Tx / 2 \right),
  \quad x\in\R^d.
\end{equation}

Current bandwidth selection techniques are either based on asymptotic expansions
\citep{Lies20} or specific model assumptions \citep{Badd16,BermDigg89,Load99}
In a recent paper, \citet{CronLies18} proposed a non-parametric alternative 
based on the Campbell--Mecke formule \citep[p.~130]{SKM} applied to the function 
$f: \R^d \to \R^+$,  $f(x) = 1\{ x \in W \} / \lambda(x)$ known as the 
Stoyan--Grabarnik statistic \citep{StoyGrab91}, which is measurable if 
$\lambda(x) > 0$ for all $x \in W$. Indeed
\begin{equation}
\label{HamiltonPrinciple}
\EE\left\{ \sum_{x\in\Psi\cap W} \frac{1}{\lambda(x)}  \right\} 
=
\int_{W}\frac{1}{\lambda(x)} \lambda(x) \, \de x
=
\ell(W).
\end{equation}
To select a bandwidth, one may simply replace $\lambda$ by an estimator 
$\widehat \lambda(\cdot; h, \Psi, W)$ in the left hand side of the equation and 
minimise the discrepancy between $\ell(W)$ and the sum of the $\widehat\lambda(x; 
h, \Psi, W)^{-1}$ over points in $\Psi\cap W$. Formally, set
\[
\label{e:HamFun}
T_{\kappa}(h;\Psi, W) = \left\{ \begin{array}{ll}
  \displaystyle
  \sum_{x\in\Psi\cap W} \frac{1}{\widehat\lambda(x;h, \Psi, W)},
& \Psi\cap W \neq \emptyset, \\
\ell(W), & \text{ otherwise},
\end{array} \right.
\]
and choose bandwidth $h>0$ by minimising 
\begin{equation}
\label{e:DefHam}
F_{\kappa}(h;\Psi, W, \ell(W)) =
\left| T_{\kappa}(h;\Psi, W) -\ell(W) \right|
\end{equation}
Since $W$ is assumed to be open and $\kappa(0) > 0$ for the 
kernels considered in this paper, $T_\kappa$ and therefore
\eqref{e:DefHam} is well-defined with or without edge correction.
Moreover, without edge correction, a zero point of the 
equation (\ref{e:DefHam}) exists.

Our goal in the next section is to extend the ideas outlined above
to adaptive bandwidths.

\section{An adaptive bandwidth selection algorithm}
\label{S:adaptive}

For patterns that contain dense as well as sparse regions, a global
`one size fits all' approach to bandwidth selection may not be suitable. 
Indeed, by definition, it leads to a compromise choice that may 
be too large for regions that contain many points and too small for 
regions with few points. The resulting estimator therefore tends to 
oversmooth and miss fine details in denser regions and contain spurious 
bumps in the sparser regions. To overcome such problems, in the context 
of random variables, \citet{Abra82a} proposed to scale the bandwidth 
in proportion to a power of the intensity function. In the point 
pattern setting, a similar adaptive kernel estimator 
\citep{DaviFlynHaze18,Lies21} is defined as
\begin{equation} 
\label{e:Abramson}
\widehat\lambda_A( x_0; h, \Psi, W )  =
  \sum_{y\in \Psi\cap W}  \frac{1}{c(y)^d h^d} \kappa\left(
    \frac{x_0-y}{h \, c(y) } \right) w(y, h, W)^{-1}
\end{equation}
where
\begin{equation}
\label{e:Abra-c}
c(y) = \left( 
  \frac{\lambda(y)}{ \prod_{z\in\Psi\cap W} \lambda(z)^{1/N(\Psi\cap W)} }
\right)^{\alpha},
\end{equation}
$N(\Psi\cap W$ denotes the number of points of $\Psi$ that fall in $W$
and $w(y,h, W)$ is an edge correction weight. The power $\alpha$ is set to 
$-1/2$ when considering asymptotic expansions \citep{Abra82a,Lies21}. 
In practice, other powers, e.g.\ $\alpha = -1/d$ when $d\geq 2$, may 
perform as well.

Let us make a few observations. First, note that points $y$ located in 
regions with a low intensity are given a larger bandwidth $h \, c(y)$ 
than those in high intensity regions, as desired. Secondly, we must assume 
that $\lambda(y) > 0$ for each $y\in \Psi \cap W$. The normalisation by the 
geometric mean is used to obtain a dimensionless quantity for the bandwidth. 
When focussing on a single point $x_0$ \citep{Abra82a,Lies21}, one could 
normalise simply by $\lambda(x_0)$. Finally, classic edge correction 
ideas apply. For example, a local edge correction weight factor in this 
context takes the form 
\[
w(y, h, W) =  \frac{1 }{ c(y)^d  h^d} \int_W
 \kappa\left(  \frac{z-y}{ h \, c(y) }  \right) dz
\]
and is mass preserving:
\[
\int_W \widehat\lambda_A( z; h, \Psi, W ) \, dz = N(\Psi \cap W).
\]

Since the local bandwidth $h \, c(y)$ depends on the unknown intensity 
function, (\ref{e:Abramson}) cannot be calculated. A common solution 
is to estimate $c(y)$ by plugging-in a {\em pilot intensity estimator\/}
\citep{ChacDuon18,Silv86,WandJone94}. For example, one could estimate 
$\lambda(y)$ by a global bandwidth kernel estimator of the form
(\ref{e:kernel}) and set 
\[
\widehat c(y) = \left( 
  \frac{\widehat \lambda(y)}
    { \prod_{z\in\Psi\cap W} \widehat \lambda(z)^{1/N(\Psi\cap W)} }
\right)^{\alpha}.
\]
We propose to use a similar two-step approach to select an adaptive 
bandwidth. More specifically, first use the Cronie \& Van Lieshout
technique to select a global bandwidth for the pilot intensity 
estimator and plug it into (\ref{e:Abra-c}) to obtain $\widehat c(y)$. 
Then apply (\ref{e:DefHam}) to $\widehat \lambda_A$ with local bandwidths
$h \, \widehat c(y)$ and optimise over $h$. 

More formally, assume that $\Psi \cap W \neq \emptyset$ and let $\kappa$
be some kernel. Then the adaptive bandwidth selection algorithm reads as 
follows.

\begin{algo}
\label{A:localCvL}
\mbox{}

\begin{description}
\item[1.a] Choose a global bandwidth $h_g$ by minimising
\[
\left| \sum_{x \in \Psi \cap W} \frac{1}{ \widehat \lambda(x; h, \Psi, W) } - \ell(W) \right|
\]
over $h>0$ where, for $x_0 \in W$,
\[
\widehat \lambda(x_0; h, \Psi, W) = \frac{1}{h^d} \sum_{y \in \Psi \cap W} 
     \kappa\left( \frac{ x_0 - y }{ h } \right).
\]
\item[1.b] Calculate a pilot estimator
\[
\widehat \lambda_g(x; h_g, \Psi, W) = \frac{1}{h_g^d} \sum_{y \in \Psi \cap W} 
   \kappa\left( \frac{ x - y }{ h_g } \right) w^{-1}(y, h_g, W) 
\]
for each $x\in\Psi$, with local edge correction 
\[
w(y, h_g, W) = \frac{1}{h_g^d} \int_W \kappa\left( \frac{ z - y }{ h_g } \right) dz.
\]
\item[2.a] Choose an adaptive bandwidth $h_a$ by minimising
\[
\left| \sum_{x \in \Psi \cap W} \frac{1}{ \widehat \lambda_A(x; h, \Psi, W) } - \ell(W) \right|
\]
over $h>0$ where, for $x_0 \in W$,
\[
\widehat \lambda_A(x_0; h, \Psi, W)  = 
\frac{1}{h^d} \sum_{y \in \Psi \cap W} \frac{1}{\widehat c(y; h_g, \Psi, W)^d}
     \kappa\left( \frac{ x_0 - y }{ h \, \widehat c(y; h_g, \Psi, W)  } \right)
\]
with
\[
\widehat c(y; h_g, \Psi, W) =  \left( 
  \frac{\widehat \lambda_g(y; h_g, \Psi, W)}
       { \prod_{z\in\Psi\cap W} \widehat \lambda_g(z; h_g, \Psi, W)^{1/N(\Psi\cap W)} }
\right)^{-1/2}.
\]
\item[2.b] Apply local edge correction, approximating when necessary,
to calculate the final estimator.
\end{description}
\end{algo}

In selecting the bandwidth, no edge correction is applied as the 
clearest optimum is obtained that way \citep{CronLies18}. Note that to 
ensure that one never divides by zero, the intensity function estimates
must be positive for all $x \in \Psi \cap W$. A sufficient condition is that 
$\kappa(0) > 0$.

Next, we consider the continuity properties of $T_\kappa(\cdot; \Psi, W)$ 
and its limits as the bandwidth approaches zero and infinity. For the global 
case, \citep[Thm~1]{CronLies18} guarantees the validity of the first step in 
the above algorithm. For step {\bf [2.a]} the following theorem holds.

\begin{thm} 
\label{t:continuous}
Let $\psi$ be a locally finite point pattern of distinct points in $\R^d$ 
observed in some non-empty open and bounded window $W$ such that 
$\psi \cap W \neq \emptyset$. Let $\kappa$ be a Gaussian kernel or a 
Beta kernel with $\gamma > 0$, and $w \equiv 1$. Write $\widehat \lambda_A$ 
for the Abramson estimator (\ref{e:Abramson}) with 
\[
c(y; \psi, W) = \left( 
   \frac{1}{\lambda_p(y)} \prod_{z\in \psi \cap W} \lambda_p(z)^{1/N(\psi\cap W)}
\right)^{1/2}
\]
for some pilot estimates $\lambda_p(y)$ that are strictly positive for all 
$y \in \psi \cap W$. Then the criterion function 
\[
T_\kappa(h; \psi, W) = \sum_{x\in\psi\cap W} 
  \frac{1}{ \widehat \lambda_{A}(x; h, \psi, W)}
\]
is a continuous function of $h$ on $(0,\infty)$. For the box kernel it is piecewise 
continuous. In all cases, 
\[
\lim_{h\to 0} T_\kappa(h; \psi, W) = 0; \quad 
\lim_{h\to\infty} T_\kappa(h; \psi, W) = \infty.
\]
\end{thm}

\begin{proof}
We will first look at the limit as $h\to 0$. Note that for all $h>0$ 
and $x\in \psi \cap W$,
\[
\widehat \lambda_A(x; h, \psi, W) \geq \kappa(0) \, c(x; \psi, W)^{-d} h^{-d}
 > 0.
\]
Here we use that since the pilot estimator $\lambda_p$ is strictly positive
on the non-empty pattern $\psi \cap W$, so is $c(\cdot)$. Also, for all kernels 
considered, $\kappa(0) > 0$.  Consequently,
\[
T_\kappa(h; \psi, W)  =  \sum_{x\in \psi \cap W} 
\frac{1}{\widehat \lambda_A(x; h, \psi, W)}
 \leq  
\sum_{x\in \psi \cap W}
\frac{c(x; \psi, W)^d h^d}{ \kappa(0)}.
\]
The right-most expression and therefore $T_\kappa(h; \psi, W)$ tends to $0$.

\medskip

Next let $h\to \infty$. For the box, Beta and Gaussian kernels, $\kappa( \cdot ) 
\leq \kappa(0)$. We already observed that $c(y; \psi, W)$ is strictly positive 
for $y\in \psi\cap W$ since by assumption $\lambda_p$ is. Moreover, it does not 
depend on $h$ and therefore
\begin{eqnarray*}
T_\kappa(h; \psi, W) & = & \sum_{x\in \psi \cap W} \frac{ h^d}{
\sum_{y\in \psi \cap W} c(y; \psi, W)^{-d} 
\kappa\left( \frac{ x - y }{ h \, c(y; \psi, W) } \right)} \\
& \geq & h^d 
\sum_{x\in \psi \cap W} \frac{1}{
\sum_{y\in \psi \cap W} c(y; \psi, W)^{-d}  \kappa(0) }.
\end{eqnarray*}
The right-most expression and therefore $T_\kappa(h; \psi, W)$
tends to $\infty$ when $\psi \cap W$ is non-empty.

\medskip

It remains to look at continuity properties. Both the Beta kernels
$\kappa^\gamma$ with $\gamma > 0$ and the Gaussian kernel are 
continuous on $\R^d$. The box kernel is discontinuous on the unit 
disc $\partial B(0,1)$ only. The function $h\to h^{-d}$ is continuous on 
$(0,\infty)$. Therefore, for fixed $z\in\R^d$, the function
\(
h \to \kappa( z /  h )
\)
is also continuous when $\kappa$ is a Gaussian kernel or a Beta
kernel with $\gamma >0$. For the box kernel, this function is 
piecewise continuous, having a discontinuity at $h=||z||$. 
Observe that, since $\psi\cap W$ is non-empty by assumption,
$\kappa(0) > 0$ and the pilot estimates $\lambda_p(x)$ are
strictly positive for every $x\in\psi\cap W$, also the
$\widehat \lambda_A(x; h, \psi, W)$ are strictly positive for 
$x\in\psi\cap W$ and $h>0$. We conclude that, as a function 
of $h$ on $(0,\infty)$, $T_\kappa(h; \psi, W)$ is continuous
for Gaussian kernels and Beta kernels with $\gamma > 1$, 
piecewise continous for the box kernel.
\end{proof}

Theorem~\ref{t:continuous} implies that Step {\bf{[2.a]}} in 
Algorithm~\ref{A:localCvL} is solvable for $h_a$. Usually, but 
not always, the solution is unique. In case of multiple solutions,
one may pick the smallest.

\section{Numerical evaluations}
\label{S:numerics}

In this section we investigate the performance of the proposed local
bandwidth selection algorithm for simulated and real-life data.

\subsection{Simulation study}
\label{S:simu}

\begin{table}[bht]
\begin{center}
\begin{tabular}{|l|l|l|}
\hline
constant & trend & high contrast feature \\
\hline
$\lambda_1(x,y) \equiv  50$
&
$\lambda_3(x,y)  =  5 + 225 \, x^4$  
&
$ \lambda_7(x,y) =  5 +  { 45 \times 50 } \times 1_S(x,y) / \pi$
\\
$\lambda_2(x,y) \equiv 250$
&
$\lambda_4(x,y) = 10 + 200 \, x^4$
&
$\lambda_8(x,y) = 10 +  { 40 \times 50 } \times 1_S(x,y) / \pi$
\\
& 
$\lambda_5(x,y) = 25 + 1125 \, x^4$
&
$\lambda_9(x,y) = 25 +  { 225 \times 50} \times 1_S(x,y) / \pi$ 
\\
& 
$\lambda_6(x,y) = 50 + 1000 \, x^4$
&
$\lambda_{10}(x,y) = 50 +  { 200 \times 50} \times 1_S(x,y) / \pi$\\
\hline
\end{tabular}
\end{center}
\caption{Intensity functions on the unit square.
Here $S = \{ (x,y) \in [0,1]^2 : (x - 0.5)^2 + (y-0.6)^2  < 1/100 
\mbox{ or } (x - 0.5)^2 + (y-0.4)^2  < 1/100  \}$.}
\label{T:intensities}
\end{table}

To compare the performance of the adaptive bandwidth selection 
approach with the global one, we conduct a simulation study.  
We consider three types of intensity functions on the unit square
in $\R^2$: a constant intensity, a gradual polynomial trend in the 
horizontal direction and a central high intensity region contrasting 
with a low intensity background. Specifically, set
\[
\lambda(x,y) = \left\{ \begin{array}{l}
 \lambda \\
 a + b x^4 \\
 a + b \, 1\{ (x,y) \in S \} 
\end{array} \right.
\]
for $S = \{ (x,y) \in [0,1]^2 : (x - 0.5)^2 + (y-0.6)^2  < 1/100 
\mbox{ or } (x - 0.5)^2 + (y-0.4)^2  < 1/100 \}$.
For each type of function, we set the parameters in such a way
that realisations contain approximately $50$ or $250$ points. For the
latter two function types, we also vary the fraction $a/b$. Doing so, 
we obtain the intensity functions summarised in Table~\ref{T:intensities}.

A convenient way to obtain realisations of point processes with spatially
varying intensity function $\lambda_i$ is to apply independent thinning 
to a realisations of a stationary point process whose intensity function
is known explicitly. Here we choose a Poisson process, a Mat\'ern cluster
process and a Mat\'ern hard core process \citep{Mate86}. We will need the 
notation $\bar \lambda_i = \sup_{(x,y) \in [0,1]^2} \lambda_i(x,y)$ for the 
maximal value of $\lambda_i$ in $[0,1]^2$.

\begin{table}[thb]
\begin{center}
\begin{tabular}{|l|rrrrr|}
\hline
$\lambda$ & Poisson & cluster $\nu=5$ & cluster $\nu=10$ & hard core $\nu=0.9$ & hard core $\nu=0.5$ \\
\hline
$\lambda_1$ & 10.22 & 23.17 & 27.96 & 10.40 & 7.73 \\
$\lambda_2$ & 31.76 & 63.77 & 82.43 & 28.37 & 24.93 \\
$\lambda_3$ & 21.99 & 33.64 & 41.34 & 21.05 & 20.16 \\
$\lambda_4$ & 16.98 & 30.51 & 43.12 & 16.31 & 15.00 \\
$\lambda_5$ & 50.57 & 71.92 & 102.96 & 48.62 & 47.40 \\
$\lambda_6$ & 39.93 & 72.06 & 99.85 & 39.69 & 35.13 \\
$\lambda_7$ & 562.61 & 565.66 & 569.71 & 561.88 & 561.03 \\
$\lambda_8$ & 434.81 & 437.03 & 441.03 & 433.15 & 433.48 \\
$\lambda_9$ & 2,805.35 & 2,801.78 & 2,804.41 & 2,801.64 & 2,800.81 \\
$\lambda_{10}$ & 2,164.57 & 2,165.48 & 2,176.19 & 2,174.04 & 2,143.22 \\
\hline
\end{tabular}
\end{center}
\caption{Mean integrated squared error relative to expected number of 
points of kernel estimates over $100$ simulations using a Gaussian kernel
with local edge correction and bandwidth chosen by the Cronie--Van Lieshout 
algorithm for different point process models having intensity functions 
$\lambda_i$, $i=1, \dots, 10$.}
\label{T:classic}
\end{table}

\paragraph{Poisson process} 
Let $X$ be a homogeneous Poisson process with intensity function
$\bar \lambda_i$. Then its independent thinning with retention probabilities 
$\lambda_i(x,y) / \bar \lambda_i$ is a heterogeneous Poisson process with 
intensity function $\lambda_i$.

\paragraph{Mat\'ern cluster process}
Let $X_p$ be a homogeneous Poisson process with intensity $\kappa$ 
on $[-0.05,$ $ 0.05]^2$. Assume that each `parent' point $z\in X_p$ generates
a Poisson number of `daughter' points, say with mean $\nu$ in the closed ball
$B(z, 0.05)$ of radius $0.05$ around $z$ and write $X$ for the union of
daughter points falling in $[0,1]^2$. Then $X$ is homogeneous and 
has constant intensity $\kappa \nu$ on $[0,1]^2$. We will consider two 
degrees of clustering:
\begin{itemize}
\item parent intensity $\kappa = \bar \lambda_i/5$, mean number of daughters 
$\nu = 5$ in a ball of radius $0.05$ around the parent;
\item parent intensity $\kappa = \bar \lambda_i/10$, mean number of daughters 
$\nu = 10$ in a ball of radius $0.05$ around the parent.
\end{itemize}
In either case, independent thinning with retention probabilities 
$\lambda_i(x,y) / \bar \lambda_i$ results in a point process $X$ having 
intensity function $\lambda_i$.

\paragraph{Type II Mat\'ern hard core process}
Let $X_g$ be a homogeneous Poisson process with intensity $\kappa$ 
on $[-r, r]^2$ and assign each `ground' point $z\in X_g$ a mark 
according to the uniform distribution on $(0,1)$ independently of 
other points. Keep a point $z\in X_g\cap [0,1]^2$ if no other point 
of $X_g$ with a larger mark lies within distance $r > 0$. The resulting 
point process $X$ is homogeneous and has constant intensity 
$(1 - e^{-\kappa \pi r^2} ) / ( \pi r^2 ) $ on $[0,1]^2$. 
We will consider two degrees of repulsion:
\begin{itemize}
\item ground intensity $\kappa = -10 \bar \lambda_i \log\nu$ with
$\nu = 0.9$ and hard core distance $r = (10 \pi \bar \lambda_i)^{-1/2}$; 
\item ground intensity $\kappa = -2 \bar \lambda_i \log\nu$ with 
$\nu = 0.5$ and hard core distance $r = (2 \pi \bar \lambda_i)^{-1/2}$.
\end{itemize}
In both cases, independent thinning with retention probabilities 
$\lambda_i(x,y) / \bar \lambda_i$ results in a point process $X$ having 
intensity function $\lambda_i$.

\begin{table}[thb]
\begin{center}
\begin{tabular}{|l|rrrrr|}
\hline
$\lambda$ & Poisson & cluster(5) & cluster(10) & hard core $\nu=0.9$ & hard core $\nu=0.5$ \\
\hline
$\lambda_1$ & 15.72 & 40.52 & 42.99 & 16.06 & 11.00 \\
$\lambda_2$ & 52.13 & 108.20 & 140.90 & 47.58 & 34.67 \\
$\lambda_3$ & 25.58 & 58.76 & 81.97 & 24.96 & 21.35 \\
$\lambda_4$ & 25.39 & 77.66 & 115.98 & 25.96 & 19.15 \\
$\lambda_5$ & 90.84 & 196.21 & 292.96 & 81.76 & 67.02\\
$\lambda_6$ & 77.80 & 188.12 & 289.79 & 81.04 & 61.98 \\
$\lambda_7$ & 555.42 & 554.96 & 560.86 & 555.32 & 555.95 \\
$\lambda_8$ & 401.12 & 403.04 & 421.07 & 396.93 & 406.10 \\
$\lambda_9$ & 2,663.56 & 2,586.62 & 2,545.90 & 2,606.78 & 2,535.12 \\
$\lambda_{10}$ & 1,731.39 & 1,828.45 & 1,799.93 & 1,717.17 & 1,808.46 \\
\hline
\end{tabular}
\end{center}
\caption{Mean integrated squared error relative to expected number
of points of kernel estimates over $100$ simulations using a Gaussian
kernel with local edge correction and bandwidth chosen by
Algorithm~\ref{A:localCvL} for different point process models having
intensity functions $\lambda_i$, $i=1, \dots, 10$.}
\label{T:Abramson}
\end{table}

The results of the simulation study are presented in Tables~\ref{T:classic}
and \ref{T:Abramson}. For each intensity function and each point process model, 
we generated $100$ simulations in the unit square and calculated the optimal 
global and adaptive bandwidths using a Gaussian kernel. The tabulated values 
are the mean integrated squared errors after local edge correction over the 
$100$ patterns scaled by the exptected number of points.  All calculations 
were done in the R-package {\tt spatstat} \citep{Badd16} to which we contributed
the function {\tt bw.CvL.adaptive}. 

Comparing Table~\ref{T:classic} to Table~\ref{T:Abramson}, for homogeneous
point processes (intensity functions $\lambda_1$ and $\lambda_2$) the mean
integrated squared error per point is smaller for a global bandwidth. This
is not surprising, as all regions are equally rich in points in expectation.
When the intensity function is increasing gradually (intensity functions 
$\lambda_i$ for $i=3, \dots, 6$) also global bandwidth selection outperforms 
adaptive bandwidth selection. The situation is reversed when the intensity 
function shows more distinct features (intensity functions $\lambda_i$ for $i=7, 
\dots, 10$). Then, for all point process models considered, local bandwidth 
selection results in a smaller mean integrated squared error relative to
the expected number of points.

\subsection{Illustration to pattern of induced earthquakes}
\label{S:Groningen}

\begin{figure}[hbt]
\epsfxsize=0.3\hsize
\epsffile{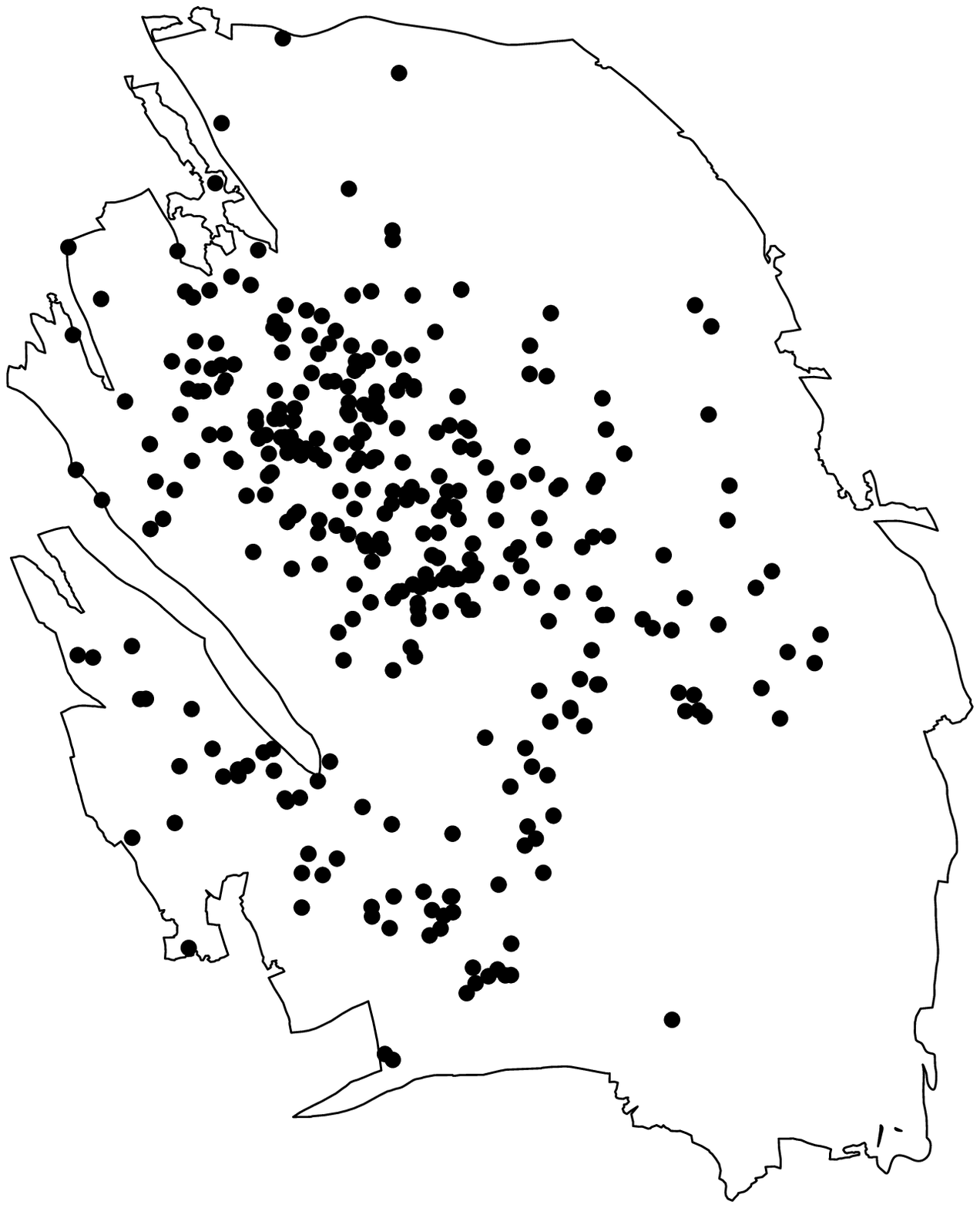}
\epsfxsize=0.3\hsize
\epsffile{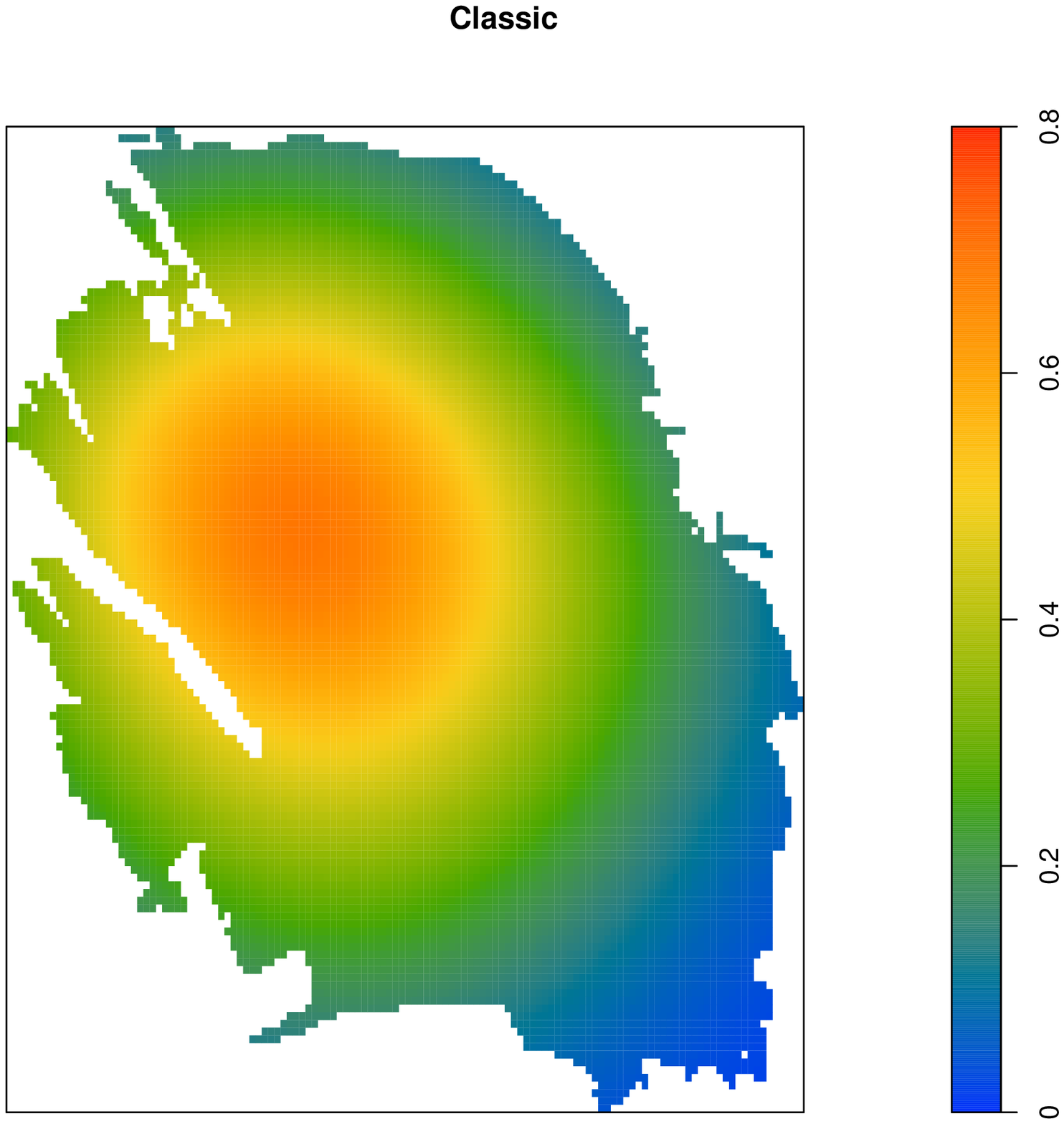}
\epsfxsize=0.3\hsize
\epsffile{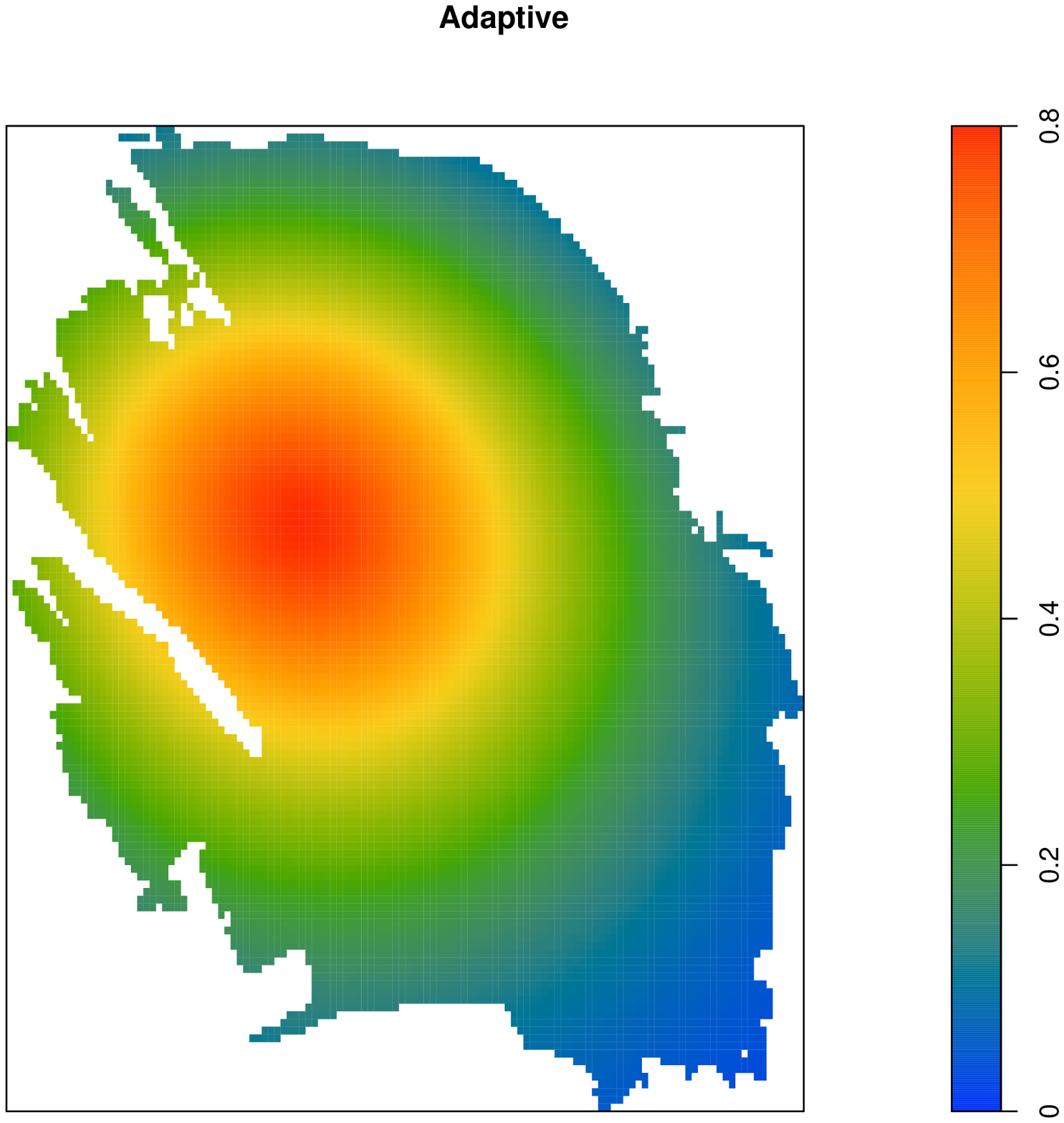}
\caption{Map of earthquakes of magnitude $M\geq 1.5$ that occurred during the
time period 1995--2021 in the Groningen gas field (left-most panel). Kernel 
estimates of the intensity function using a Gaussian kernel with local edge
correction and bandwidth selected by the Cronie--Van Lieshout algorithm (middle 
panel) and by Algorithm~\ref{A:localCvL} (right-most panel).}
\label{F:fig1}
\end{figure}

In 1959, a large gas field was discovered in Groningen, a province in
the north of The Netherlands. Initially, the benefits from the sale
of gas were a boon to the Dutch economy. However, from the 1990s 
earthquakes were being registered in the previously tectonically
inactive Groningen region. The pattern of induced earthquakes 
of magnitude $1.5$ and larger during the period 1995--2021 is depicted
in the left-most panel in Figure~\ref{F:fig1}. Note that most earthquakes 
occurred in the central and western regions.  

We applied Algorithm~\ref{A:localCvL} to produce a map of the spatially varying 
intensity function using a Gaussian kernel and local edge correction. The result 
is shown in the right-most panel in Figure~\ref{F:fig1}. For comparison, the
middle panel shows the estimated kernel estimator upon applying the Cronie--Van
Lieshout bandwidth selection algorithm. We conclude that an adaptive approach 
leads to higher estimated risks in the central gas field, balanced by a lower 
estimated risk in the periphery.

%
%
%

\section{Conclusion}

In this article, we introduced a completely non-parametric two-step 
algorithm for adaptive bandwidth selection for kernel estimators of 
the spatial intensity function and proved its validity. Simulations 
showed that for patterns with strong contrasts in point densities, 
the adaptive kernel estimator outperforms the classic kernel estimator
in terms of integrated squared error. 

We also demonstrated the feasibility of the proposed algorithm in
practice. Given a pilot estimator, the numerical complexity of step 
{\bf{[2.a]}} in Algorithm~\ref{A:localCvL} is of the same magnitude as 
that of step {\bf{[1.a]}}. Indeed, for a pattern with $n$ points, the 
calculation of $\widehat \lambda_A$ requires $n$ function evaluations per 
point. Therefore, calculation of the criterion function $F_\kappa$ is 
quadratic in $n$. Discretising the range of bandwidth values into $n_h$ 
steps, the total computational load is therefore of the order $n_h n^2$. 

Step {\bf{[2.b]}} may be computationally demanding, though. For an 
$M\times N$ grid, equation~(\ref{e:Abramson}) requires $n M N$ 
evaluations of the kernel, which may be problematic for large patterns
and fine grids. The numerical complexity of calculating the edge correction 
weights is dependent on the type of edge correction chosen ($n M N$ for local,
$M^2 N^2$ for global edge correction). When direct computation is impossible, 
fast approximation techniques exist \citep{DaviBadd18}. Note that for a global 
bandwidth, direct calculation can be avoided because (\ref{e:kernel}) may
be written as a convolution and fast Fourier techniques apply.

Finally, in this article we only considered isotropic kernels. In future, 
we plan to study adaptive kernel estimators with different bandwidths for the 
various components.

\section*{Acknowledgements}
This research was supported by The Netherlands Organisation
for Scientific Research NWO (project DEEP.NL.2018.033).

\end{document}